\documentclass[11pt]{article}
\usepackage{fontenc}
\usepackage[utf8]{inputenc}
\usepackage[margin=1in]{geometry}

\usepackage[suppress]{color-edits}
\addauthor{vs}{red}
\addauthor{jm}{blue}
\addauthor{nd}{magenta}
\usepackage{natbib}
\usepackage{mathtools}
\usepackage{fullpage}
\usepackage{amsmath,amsfonts,graphpap,amscd,mathrsfs,graphicx,lscape}
\usepackage{epsfig,amssymb,amstext,xspace}
\usepackage{algorithm,algorithmic}
\usepackage{amsthm}
\usepackage{hyperref}

\newtheorem{thm}{Theorem}[section]
\newtheorem{cor}[thm]{Corollary}
\newtheorem{lem}[thm]{Lemma}
\newtheorem{obs}[thm]{Observation}

\newtheorem{mdef}[thm]{Definition}

\newcommand{\argmax}{\textrm{argmax}}

\newcommand{\valopt}[1]{v_{#1,j^*(#1)}}
\newcommand{\jopt}[1]{j^*(#1)}

\newcommand{\state}[3]{\ensuremath{<#1,#2,#3>}}
\newcommand{\dset}[2]{\ensuremath{D(#1,#2)}}
\newcommand{\opt}{\ensuremath{\textsc{Opt}}}
\newcommand{\E}{\ensuremath{\mathbb{E}}}
\newcommand{\opti}[1]{\ensuremath{J^{*}_{#1}}}

\newcommand{\V}{\ensuremath{\mathcal{V}}}
\newcommand{\vp}{\ensuremath{{\bf v}}}

\newcommand{\set}[1]{\left\{ #1 \right\} } 
\newcommand{\reals}{\mathbb{R}} 
\newcommand{\xos} {\ensuremath{\textsc{XOS}}}
\newcommand{\Di}{\mathcal{D}_i} 

\newcommand{\ceil}[1]{\ensuremath{\left \lceil{#1}\right \rceil }}

\newcommand{\hide}[1]{}

\def\squareforqed{\hbox{\rule{2.5mm}{2.5mm}}}

\def\QED{\ifmmode\squareforqed % in mathmode : print just the square
  \else{\nobreak\hfil   % \hfil to end of current line
    \penalty50                 % penalty 50 for breaking the line here
    \hskip1em                  % leave at least 1em before the square
    \null                      % \hbox{}
    \nobreak                   % prohibit line break
    \hfil                      % another \hfil (if a break occurred)
    \squareforqed              % put the square here
    \parfillskip=0pt           % the line really ends here
    \finalhyphendemerits=0     % ignore a hyphen on the line above
    \endgraf}                  % end the paragraph
  \fi}

\def\blksquare{\rule{2mm}{2mm}}
\def\qedsymbol{\blksquare}
\newcommand{\bg}[1]{\medskip\noindent{\bf #1}}
\newcommand{\ed}{{\hfill\qedsymbol}\medskip}

\newenvironment{proofof}[1]{{\it{Proof of #1 : }}}{\ed}

\usepackage[usenames,dvipsnames]{xcolor}
\usepackage{tikz}
\usetikzlibrary{positioning}

\usepackage{times}

\begin{document}
\title{Draft Auctions}
\author{
Nikhil R. Devanur\thanks{Microsoft Research, {\tt nikdev@microsoft.com}} 
\and 
Jamie Morgenstern\thanks{Carnegie Mellon University,  {\tt jamiemmt@cs.cmu.edu}}
\and
Vasilis Syrgkanis\thanks{Cornell University, {\tt vasilis@cs.cornell.edu}, Supported in part by Simons Graduate Fellowship in Theoretical Computer Science. Part of work done while an intern at Microsoft Research, New England.}}

\maketitle
\begin{abstract}
% !TEX root=draft.tex

We introduce draft auctions, which is a sequential auction format where at each iteration players bid for the right to buy items at a  fixed price.  We show that draft auctions offer an exponential improvement in social welfare at equilibrium over sequential item auctions where predetermined items are auctioned at each time step. Specifically, we show that for any subadditive valuation the social welfare at equilibrium is an $O(\log^2(m))$-approximation to the optimal social welfare, where $m$ is the number of items.  We also provide tighter approximation results for several subclasses. Our welfare guarantees hold for Bayes-Nash equilibria and for no-regret learning outcomes, via the smooth-mechanism framework.  Of independent interest, our techniques show that in a combinatorial auction setting, efficiency guarantees of a mechanism via smoothness for a very restricted class of {\em cardinality} valuations, extend with a small degradation, to subadditive valuations, the largest complement-free class of valuations.  Variants of draft auctions have been used in practice and have been experimentally shown to outperform other auctions. Our results provide a theoretical justification.
\end{abstract}
\thispagestyle{empty}
\newpage
\setcounter{page}{1}
% !TEX root=draft.tex
\section{Introduction}\label{sec:intro}

Consider the scenario where several indivisible items are to be
auctioned off to bidders with {\em combinatorial} valuations, i.e.,
valuations that depend on the entire set of items obtained.  In
practice, simple auctions such as {\em sequential item auctions} are
commonly used for such purposes.  It is known that these auctions
could lead to a highly inefficient allocation of items, as measured by
the {\em price of anarchy},
% . shown to have a bad price of anarchy with respect to the optimal social welfare, 
even for very simple valuation functions and a {\em complete
  information} setting.  We introduce a natural and simple alternative
called a {\em draft auction} which has a much (exponentially) better
price of anarchy for the very general class of subadditive valuation
functions, for the {\em incomplete information} setting.  We discuss
why this makes a strong case for adopting draft auctions in place of
sequential item auctions in practice.

An {\em instance} of the scenario we wish to study consists of $m$
{\em items} that are to be auctioned off, $n$ {\em bidders} wishing to
obtain these items, and a {\em valuation} function $v_i:2^{[m]}
\rightarrow \reals_+$ for each bidder $i$.  (We identify the set of
items and the set of bidders with $[m]$ and $[n]$ respectively.)  We
assume that the $v_i$s are monotone and non-decreasing.  The {\em result}
of an auction is an allocation of items to bidders and payments of
bidders: bidder $i$ gets a set $S_i \subseteq [m]$ of items, and makes
a payment $P_i$, with the $S_i$s forming a partition of $[m]$.
Bidders are selfish and try to maximize their utility from the
auction, which is assumed to be {\em quasi-linear}, i.e., $v_i(S_i) -
P_i$.  The valuation $v_i(S)$ can then be interpreted as how much the
set of items $S$ is worth to $i$, in terms of the {\em numeraire} in
which the payments are made.  Suppose that the objective of the
auction designer is to maximize the {\em social welfare} of the
resulting allocation, which is defined as $SW := \sum_{i \in [n]}
v_i(S_i)$.

An example of such an auction that is commonly seen in practice is
what is called a {\bf sequential item auction}: items are auctioned
off one after the other (in some arbitrary order), using a simple
auction such as an ascending price auction or a sealed bid first or
second price auction.  To be precise, consider a sequential,
sealed-bid first price auction which is formally defined as follows.
There are $m$ rounds, and in each round $j \in [m]$ each bidder
$i\in[n]$ submits a bid $b_{ij}.$ Item $j$ is sold to the highest
bidder $i^* = \arg \max_{i \in [n]} \{ b_{ij} \} $, at the price equal
to her bid, $b_{i^*j}$, breaking ties arbitrarily. The winner's
identity $i^*$ and the winning bid $b_{i^*j}$ are publicly revealed
before proceeding to the next round.
%\begin{enumerate}
%\setlength{\itemsep}{0pt}
%\item Initialize, for all $i\in [n], ~S_i = \emptyset, p_i = 0.$
%\item For $j= 1..m$,
%\item \hspace{1cm} Each bidder $i\in[n]$ submits a bid $b_{ij}.$
%\item \hspace{1cm} Allocate item $j$ to $i^* = \arg \max_{i \in [n]} \{ b_{ij} \} $, i.e., $S_{i^*} = S_{i^*} \cup \{j\}$.  Break ties arbitrarily.
%\item \hspace{1cm}  Bidder $i^*$ pays her bid, i.e., $p_{i^*} = p_{i^*} + b_{i^*j}$. 
%\item End For. 
%\end{enumerate} 

Notice that the allocation of items in this auction is a function of
the bids, and each bidder strategizes to maximize her own utility.
The bid of a bidder in any round could be a function of her own
valuation, the information the bidder has about other bidders'
valuations, and the observed history until that time, which includes
the winners and their bids in all previous rounds.  In general there
is no single utility-maximizing strategy for a bidder since her
utility also depends on other bidders' strategies, thus setting up a
{\em game} among the bidders.

Rational players are assumed to play {\em equilibrium} strategies,
where each bidder's strategy is a ``best response'' to the strategies
of all the other bidders. There are many equilibrium definitions in
the same spirit as above but differing in technical details; see
Section \ref{sec:prelims} for precise definitions.

\paragraph {Bounding the inefficiency at equilibrium via the price
  of anarchy.}
Equilibria of certain auctions lead to allocations that are not
welfare optimal. It is standard practice to analyze this inefficiency
by bounding the ratio of welfares of the optimal allocation and the
welfare-minimizing equilibrium of the auction. Such a bound is called
the {\bf Price of Anarchy} ($PoA$). The { Price of Anarchy} provides a
quantitative scale with which we can measure such auctions;\footnote{Analogous to an approximation factor for approximation algorithms or a competitive ratio for online algorithms.} a smaller
price of anarchy is more desirable. To be precise, for a given
valuation profile $\vp$, let $SW(\opt(\vp))$ be the optimal social
welfare, which is the highest social welfare obtainable over all
possible allocations of items to bidders.  $SW(\opt(\vp)) := \max
\set{ \sum_{i \in [n]} v_i(S_i): (S_i)_{i\in [n]} \text{ is a
    partition of } [m] }$. Let $T$ denote a particular set of
equilibria, $s$ an equilibrium in $T$ and $SW(s)$ the social welfare
at this equilibrium. Then
\[ PoA(T) := \max_{s\in T}\frac{SW(\opt(\vp))}{SW(s)}.\]

The price of anarchy defined above is for a given instance; it can be
generalized to a {\em Bayesian} setting, which formalizes the notion
that bidders have probabilistic beliefs about each others
valuations:\footnote{What we call the {\em Bayesian} setting here is
  also called the {\em incomplete information} setting.}  each $v_i$
is drawn independently from a probability distribution $\Di$ for all
$i\in [n]$.  The $\Di$s are public knowledge, but $v_i$ is bidder
$i$'s private information.  $\Di$ represents the belief about bidder
$i$'s valuation based on publicly available information.  The price of
anarchy is then defined as a ratio of expectations, expectation of
$SW(\opt)$ and expectation of $SW(s)$.  The expectations are taken
over the draws of $v_i$ from $\Di$ for each $i$.  The {\em complete
  information} setting where all bidders know all valuations is a
special case of the Bayesian setting.

The price of anarchy of an auction can crucially depend on the
structure of the valuation functions; therefore, we consider special
classes of valuation functions and study the worst-case (maximum)
price of anarchy over all instances with valuation functions belonging
to each class. Among the simplest valuation function classes are {\em
  additive} valuations, which are of the form $v_i(S_i) = \sum_{j \in
  S_i} v_{ij} $ and {\em unit-demand} valuations which are of the form
$v_i(S_i) = \max_{j \in S_i} \{v_{ij}\} .$ That is, a unit-demand
bidder values a bundle only according to his most-valued item in the
bundle.

It was recently shown by \cite{Feldman2013b} that for sequential first
price auctions, when bidders may have {\em either} additive or
unit-demand valuations, {\bf the price of anarchy could be $\Omega(m)$
} for the set of pure Nash equilibria in the complete information
setting.\footnote{See Section~\ref{sec:prelims} for formal definitions
  of equilibria and the complete information setting.}  Since the
class of additive/unit-demand valuations are among the simplest
valuations and the set of pure Nash equilibria in the complete
information setting is among the smallest set of equilibria, the price
of anarchy for this case should be among the lowest.  Yet the lower
bound of $\Omega(m)$ is nearly as bad as it gets since it is easy to
show an upper bound of $O(m)$ for a much more general class of
valuations (subadditive valuations) and a much bigger set of
equilibria.
%%% We have not said unit demand has constant PoA. 
%\footnote{Each bidder get value at least
%  $\frac{1}{m}$ of what she got at $OPT$ from a single item; the proof
%  then follows the proof for the unit-demand case.} 

%so this
%example is tight.  This lower bound suggests we look beyond sequential
%first price auctions for approximate efficiency in sequential
%settings.

\subsection*{Our Contributions}
We propose a natural and simple variant of the sequential item auction
which we call a {\bf draft auction}.  Draft auctions also proceed in
rounds: each round is a sealed-bid first price auction. The difference
is that there is no designated item in any round; instead, the winner
decides which items she wishes to purchase in that round, paying her
bid for {\em each} such item. Formally, a draft auction is as follows.
\begin{enumerate}
\setlength{\itemsep}{-3pt}
\item Initialize, for all $i\in [n], ~S_i = \emptyset, P_i = 0.$ The set of remaining items $I=[m]$. 
\item While $I \neq \emptyset$,
\item \hspace{1cm} Each bidder $i\in[n]$ submits a {sealed} bid $b_{i}$ and a set $X_i \subseteq I.$
\item \hspace{1cm} Allocate set $X_{i^*} $ to $i^* = \arg \max_{i \in [n]} \{ b_{i} \} $, i.e., $S_{i^*} = S_{i^*} \cup X_{i^*} $.  Break ties arbitrarily.
\item \hspace{1cm}  Bidder $i^*$ pays her bid for each item in $X_{i^*} $, i.e., $P_{i^*} = P_{i^*} + b_{i^*} |X_{i^*} |$. 
\item {\hspace{1cm}  The winner $i^*$, winning bid $b_{i^*}$ and allocated bundle $X_{i^*}$ is announced. }
\item End While. 
\end{enumerate} 

We show that draft auctions have a much better price of anarchy than
sequential item auctions, for the very general class of subadditive valuation
functions.  {\bf Subadditive} valuations are those $v$ that satisfy
the property $v(S \cup T) \leq v(S) + v(T)$ for all $S, T\subseteq
[m]$.  The class of subadditive valuations, which are also called
complement-free valuations, contains other
well-studied classes of valuations such as submodular, gross
substitutes (see Appendix \ref{app:valuations} for formal
definitions), additive and unit-demand valuations.  We show the
following price of anarchy bound for draft auctions for subadditive
valuations.
\begin{thm}\label{thm:subadditive}
  The price of anarchy for draft auctions for subadditive valuations
  with respect to Nash equilibria (Definition \ref{def:nash}) in the
  Bayesian setting or correlated equilibria (Definition
  \ref{def:correlated}) in the complete information setting is
  $O(\log^2 m)$.
\end{thm} 
We show a slightly better bound for the class
of $\xos$ valuations, which is the class of valuations that are representable
as a maximum of linear functions, i.e., valuations of the form
\[ v(S) = \max\set{\sum_{j\in S}v_{1j}, \ldots, \sum_{j\in
  S}v_{kj}}.\] 
\begin{thm}\label{thm:xos}
  The price of anarchy for draft auctions for $\xos$ valuations with
  respect to Nash equilibria in the Bayesian setting or correlated
  equilibria in the complete information setting is $O(\log m)$.
\end{thm}
The relations between these classes of valuations are given below. 
\[ \text{unit-demand } \cup \text{ additive } \subseteq \text{ gross
  substitutes } \subseteq \text{ submodular } \subseteq \xos \subseteq
\text{ subadditive} .\]

We also show constant factor upper and lower bounds for the price of
anarchy for unit-demand valuations as well as for symmetric concave
valuations (where the valuation is a concave function of only the {\em
  number} of items; see Section \ref{sec:extension} for a precise
definition).
\begin{thm}\label{thm:unit-demand}
  The price of anarchy for draft auctions for unit demand bidders with
  respect to Nash equilibria in the Bayesian setting or correlated
  equilibria in the complete information setting is at most 4, and
  w.r.t. pure Nash equilibria in the complete information setting is
  at most 2.
\end{thm} 

\begin{thm}\label{thm:lowerbounds}
  The price of anarchy for draft auctions for unit demand bidders
  w.r.t. pure Nash equilibria in the complete information setting is
  at least 1.22. Further there are instances where no equilibrium
  achieves a welfare within $1+\epsilon$ of the optimum, for some
  small universal constant $\epsilon>0$.
\end{thm} 

\begin{thm}\label{thm:symmetric}
  The price of anarchy for draft auctions for bidders with symmetric
  concave valuations with respect to Nash equilibria in the Bayesian
  setting or correlated equilibria in the complete information setting
  is at most 8.
\end{thm} 
The price of anarchy bounds we show are exponentially better than
those for sequential item auctions.  In fact, it is possible that
draft auctions have a constant price of anarchy for subadditive
valuations.  We use this contrast to advocate the use of draft
auctions in place of sequential item auctions in practice.  

To prove our upper bounds, we use the smoothness approach introduced
by \cite{Roughgarden2009} and extended to auctions by
\cite{Syrgkanis2013}. It boils down to the following main technique:
for every equilibrium, construct a deviating strategy for each player
which gets at least some fraction of her value in the social-welfare
maximizing allocation, while paying at most a small multiple of the
revenue in equilibrium. The deviations we construct are more involved
than those for sequential item auctions and the technical difficulties
involved are detailed in Section \ref{sec:smoothness}.
% Contrary to sequential item auctions the deviation constructed here
% is rather involved.  The first main technical difficulty is that
% there is no designated auction or iteration at which any item is
% sold. Hence, the deviations must attempt to win at all rounds.  The
% second difficulty is common to all auctions/games where the bidders
% persist across multiple rounds: a deviation in one round may take
% the auction along an off-equilibrium path about which it is very
% hard to argue.  We get around this difficulty by restricting the
% deviations to deviate from the equilibrium path in only a single
% round.  the same time are not affecting the equilibrium path and
% altering future prices.
On a separate note, we show that efficiency bounds proven via the
smoothness approach for a very special class of valuations directly
extend with only a polylogarithmic degradation to the whole class of
subadditive valuations and with no degradation to the class of
symmetric concave valuations. Specifically, we show that it suffices
to analyze settings where the value of a player is simply proportional
to the number of items he acquired from a specific interest set of
items. Then we show that smoothness for these simple constrained,
cardinality valuations directly implies smoothness for concave
symmetric valuations (i.e. identical items) with no loss, for
submodular valuations with only a $\log(m)$ loss and for subadditive
valuations with a $\log^2(m)$ loss. Our approach may have potential
applications to the analysis of other simple mechanisms for
combinatorial auction settings.

%The proofs of the upper bounds in this paper rely on one
%main technique: construct a deviation for each player which gets at
%least some fraction of their value according to the social-welfare
%maximizing allocation, while paying at most some fraction of the
%revenue in equilibrium. The deviations we construct are
%\emph{one-shot}: once the deviation wins some round, the bidder buys a
%set of items and then bids 0 in the remaining rounds.  This technique
%is used to circumvent reasoning about off-equilibrium paths: once a
%bidder makes a noticeable deviation in one round, the behavior of
%other bidders in following rounds may change substantially. This
%subtlety is also the root of our need to allow bidders to take sets of
%items once they've won, rather than only allowing for a single item to
%be taken in each round. For the same reason, we assume the losers of a
%round learn only the identity of the winner, the items they selected,
%and their bid, rather than the bids of all participants.

%\section{Illustrative Example and Related Issues} 
%We now given an illustrative example and then discuss related work. 
%several related issues and questions.
% that could further shed light on which is the better auction format.  

\paragraph{Right to choose auctions.}
A simpler variant of the draft auction is obtained by restricting each
bidder to only pick one item when she wins a round.  This auction
format has been studied and used previously, under the names of
``right to choose'' (RTC) auctions or ``pooled auctions''.
Intuitively, the two formats should not differ much; if a bidder wins
a round at a certain price in an RTC auction, then she should be able
to win subsequent rounds at the same price too, thus simulating a
draft auction.  The reason that our results don't readily extend to
this format is that the deviations we construct in our proofs need the
ability to win multiple items at once. We believe that this is a
technical limitation and that our price of anarchy bounds should
extend to RTC auctions as well.

\paragraph{Illustrative example.} 

\hide{\jmcomment{This isn't quite the same as the example in the SODA paper,
  where the order is ABC, and the numbers are $\alpha$ and
  $\epsilon$. I suspect it's best if we stick to the same convention
  if we're not going to fully flesh this out, so that people can refer
  to the other paper for the figure. Here is the modified version, if you agree:}
}
%\jmcomment{ 
To illustrate the advantages of draft auctions over
  sequential item auctions, we revisit an instance introduced by
  \cite{PaesLeme2012}, that shows that inefficiency is bound to arise
  at the unique subgame perfect equilibrium in undominated strategies
  of sequential item auctions with unit-demand bidders: Consider an
  instance with $4$ bidders, $a, b, c, d$ and $3$ items $A, B,
  C$. Bidder $a$ has value $v_a=\epsilon$ only for item $A$, bidder
  $b$ has value $\alpha$ for either $A$ or $B$, bidder $c$ has value
  $\alpha$ for either $B$ or $C$ and bidder $d$ has value $\alpha -
  \epsilon$ for $C$.  It is shown by \cite{PaesLeme2012} that assuming
  that auctions occur in order $A, C, B$ then in the unique
  equilibrium, bidder $b$ will let the $\epsilon$-valued bidder $a$,
  win the auction, so that he gets the last auction for item $B$ for
  free.  The reasoning being that bidder $c$ will go for item $C$ and
  will not bid in the last auction. This yields a price of anarchy of
  $3/2$.
%}

%\jmcomment{  
However, observe that the latter behavior is very much tied to the
  ordering of the item auctions. If the auctioneer were to
  run a draft auction in the same setting then it is easy to see that
  the optimal allocation can arise at equilibrium: bidders $b,c,d$ all bid
  $\epsilon^+$ at every iteration until they get allocated. If bidder
  $b$ wins he gets item $A$, if bidder $c$ wins then he gets item $B$
  and if bidder $d$ wins he gets item $C$. It is easy to see that no
  bidder has an incentive to deviate.
\paragraph{Related work.} 
% !TEX root=draft.tex
% related.tex

A predominant approach to combinatorial auctions is the design of
``truthful mechanisms''. Although the VCG mechanism is truthful and
gives the socially optimal allocation, it is not computationally
efficient. There has been a long line of research into designing
truthful mechanisms that run in polynomial time and approximate the
social welfare for various classes of valuations: see \citet{Nisanagt}.

More recently, an alternate approach has been to analyze simple
auctions that are commonly used in practice, by quantifying the
inefficiency of equilibria via the price of anarchy
\citep{Christodoulou2008, Bhawalkar2011, Hassidim2011, Feldman2013, Lucier2010, PaesLeme2010, Lucier2011, Caragiannis2011}.
Our work is most closely related to recent results on sequential item
auctions: \citet{PaesLeme2012} showed a price of anarchy of 2 for
unit-demand valuations in the complete information case, and that for
submodular bidders the price of anarchy can grow linearly with the
number of items.
%I know this paragraph is too long but there is no good breaking point. They are all closely related, so I just left it like this
The positive results were later extended to the incomplete-information
setting by \citet{Syrgkanis2012a}.  
A dominating theme here has been
the emergence of a ``smoothness'' framework that captures many of the
price of anarchy bounds, and allows these bounds to be extended to
larger classes of equilibria:
%Another interesting line of research is around putting many of the
%price of anarchy bounds into a ``smoothness'' framework and showing
%that these bounds apply to a larger class of equilibria:
\citet{Roughgarden2009} to outcomes of learning algorithms and
\citet{Roughgarden2012} and \citet{Syrgkanis2012} to games of
incomplete information.
%\citet{Roughgarden2009} showed that many price of anarchy bounds carry
%over to imply bounds also for learning outcomes.
%\citet{Roughgarden2012} and \citet{Syrgkanis2012} showed that such
%bounds also extend to bound the inefficiency of games of incomplete information. 
\citet{Syrgkanis2013} give a  specialized smoothness
framework for auctions with
quasi-linear preferences, which we also use. 
% showing how to capture several of the
%previous results. In general smoothness of a mechanism is defined via the existence
%of special strategies, that each bidder can employ and get a good fraction of his valuation
%under the optimal allocation, by not paying much more than what is being currently paid for that allocation. 
%Our upper bounds use the the framework of \citet{Syrgkanis2013}. 
In fact, we provide a way to extend the smoothness for a very simple class of valuations to 
smoothness for subadditive valuations with only a polylogarithmic loss. 
This potentially has applicability in the analysis of other simple
mechanisms for subadditive valuations.
On the negative side, \citet{Feldman2013b}
showed that even when some valuations are unit-demand and some are
additive, the price of anarchy of sequential item auctions can grow linearly with the number of items. 
Our work shows that this inefficiency can be largely alleviated by switching to the draft
auction, thereby portraying that it was not the sequentiality
that caused the inefficiency but rather the specific ordering of the
items being auctioned.

In the economics community the literature on right to choose (RTC)
auctions is the closest to our work.  Most of this work is empirical, some 
in the field and others in the lab, and shows that the revenue of RTC auctions is higher than that of other auctions. 
Among field experiments \citet{ashenfelter} studied the result of RTC auctions in condominium
sales in Miami, which indicated\footnote{We find the results inconclusive, due to reasons we cannot go into here.} that the revenue of RTC auctions could be higher than other formats. 
%37\% of the buyers of condos sold at auction
%ultimately failed to pay for the condos: the condos they had bought
%were then sold to other buyers in face-to-face bargaining. The
%per-condo revenue was substantially higher in the draft auction (13\%
%higher) than in the subsequent face-to-face bargaining.  
\citet{alevy} studied RTC auctions for 
water rights sales in Chile and found higher revenue
than in the analogous sequential item auction. 
Laboratory experiments by \citet{eliaz}, \citet{goeree} and  \citet{salmon} 
all find evidence of higher revenue in RTC auctions under various settings. 
%the revenue of the RTC auction is even larger than what theory would
%suggest the optimal mechanism should achieve, when bidders are
%single-minded. 
% show experimental evidence which
%suggests RTC auctions can generate more revenue than simultaneous item
%auctions with single-minded and risk-averse bidders.
%finds that RTC auctions can generate
%higher revenue than simultaneous ascending price auctions in a setting
%when bidders are not risk-averse. \ndcomment{Are these in the lab or in the field?} 

Most theoretical work on RTC focuses on very special cases.
\citet{harstad} finds that revenue equivalence holds between RTC and
sequential item auctions, for 2 superadditive bidders.  \citet{gale}
has shown that all Bayes-Nash equilibria yield socially optimal
allocations for 2 unit-demand bidders.  \citep{burguet} shows that RTC
generates more revenue than sequential item auctions, when there are
$2$ items and many single-minded, risk-averse bidders, each equally
likely to prefer either item, whose valuations are drawn i.i.d from a
continuous distribution.% \ndcomment{This sounds theoretical, is it? }
Yet, it is not clear if RTC auctions always generate a higher revenue
than other auctions for a general setting.

The economics literature on sequential item auctions is focused on
exact characterizations, once again for very special cases
\citep{Weber2000,Milgrom1982a}.  These become exceedingly difficult as
we go beyond a few items.

\section{Preliminaries and Notation}\label{sec:prelims}

\newcommand{\s}{\mathbf s} 
\newcommand{\sminusi}{{\mathbf s}_{-i}} 
\renewcommand{\v}{\mathbf v} 
\newcommand{\vminusi}{{\mathbf v}_{-i}} 
\newcommand{\strats}{\mathcal{S}} 
\newcommand{\stratsp}{\mathcal{S}_P} 
\newcommand{\stratsm}{\mathcal{S}_M} 

Recall that in the {\bf Bayesian setting}, each $v_i$ is drawn
independently from a distribution $\Di$ on a set of possible
valuations $\V_i$, all $\Di$s are public knowledge and $v_i$s are
private information.  In each round, the winner, the winning set and
the winning price are publicly revealed.  The {\bf complete
  information} setting is a special case where each bidder knows the
valuation of all the other bidders.\footnote{It is the case where
  $D_i$ is $v_i$ with probability 1.}
%Then the auction becomes a complete information sequential game. 

A {\em strategy} $s_i:\V_i\rightarrow \Delta(B_i)$ of bidder $i$ is a
function, from her valuation to a distribution over bid plans $b_i\in
B_i$. Each bid plan $b_i$ determines the bid $b_{it}$ that a player
makes at some round $t$ and the set $X_{it}$ of items he gets
conditional on winning, based on the information $h_{it}$ available to
her up to that round.  For any given valuation profile $\v$, a tuple
of strategies $b=\s(\v) = (s_i(v_i))_{i\in [n]}$ determines the
outcome of the auction; let $u_i(b;v_i)$ denote the utility, (or
expected utility when $b$ is a distribution over bid plans) obtained
by bidder $i$ as a function of the bid plans $b$. Recall that for a
deterministic profile the utility is $v_i(S_i(b)) -P_i(b)$ where
$S_i(b)$ is the set of items $i$ wins and $P_i(b)$ is her total
payment.  Additionally, for any bid plan $b$, we denote with $p_j(b)$
the price that item $j$ was sold at, under bid plan $b$. Observe that
a bid plan actually also contains information about {\em what might
  have happened}, i.e., they specify the result of possible deviations
from the actual outcome, which becomes important in the definitions of
equilibria.  We now define the most basic equilibrium concept, that of
a Nash equilibrium.
\begin{mdef} \label{def:nash} A {\bf pure (resp. mixed) Bayes-Nash}
  equilibrium is a pure (resp. mixed) strategy tuple $\s$ such that no
  player can unilaterally deviate to obtain a better utility. In other
  words,
  \[ \forall~i\in [n], \forall v_i\in \V_i, \forall~b_i'\in B_i,
  ~~\E_{\vminusi}[u_i (b_i',\sminusi(\v_{-i});v_i)] \leq
  \E_{\vminusi}[ u_i(\s(\v);v_i)],\] where as is standard,
  $\sminusi(\v_{-i})$ denotes $(s_j(v_j))_{j\in [n], j\neq i} $, the
  strategy tuple $\s$ restricted to players other than $i$, and
  $(b_i',\sminusi(\v_{-i}))$ denotes the tuple where $s_i(v_i)$ is
  replaced by $b_i'$ in $\s(\v)$.  Similarly $\vminusi$ denotes the
  tuple of valuations $(v_j)_{j\in [n], j\neq i} $.  The expectations
  are taken over the draw of $\vminusi$.
\end{mdef}

A Nash equilibrium in sequential games allows for {\em irrational
  threats}, where an equilibrium strategy of a bidder could be
suboptimal beyond a certain round.  A standard refinement of the Nash
equilibrium for extensive form games is the \emph{subgame perfect
  equilibrium}, that allows only for strategies that constitute an
equilibrium of any subgame, conditional on any possible history of
play (see \cite{fudenberg91} for a formal definition and a more
comprehensive treatment.) Our results also extend to
complete-information correlated equilibria, as defined in
Appendix~\ref{def:correlated}.

\[ \text{Subgame perfect } \subseteq \text{ Nash } \subseteq \text{Correlated Equilibria} \]

The price of anarchy may be defined w.r.t any of these equilibria;
larger classes have higher price of anarchy.  In the Bayesian setting
the price of anarchy is defined as the worst-case ratio of the expectations, over the random values, of the social
welfare at the optimum $\E_{\vp}[SW(\opt(\vp))]$ and at an equilibrium $\E_{\vp}[SW(\s(\vp))]$.
% \jmdelete{\begin{mdef} \label{def:poa}
%The {\bf Price of Anarchy} w.r.t a given set of equilibria $T$, for the Bayesian setting given  by the distributions $\Di$ for all $i \in [n]$  is 
%\[ PoA(T) = \max_{\s \in T}  \set {\frac {\E_{\v \sim \Pi_i\Di} [SW(\opt(\vp))] }   {\E_{\v\sim \Pi_i\Di} [ SW(\s(\vp))] } } .\] 
%\end{mdef} }

%\vsdelete{The classes of unit-demand, additive and subadditive valuations were defined in Section \ref{sec:intro}. 
%Submodular and gross substitutes valuations are defined in Appendix \ref{sec:}. Here we define the class of 
%$k$-restricted complements. 
%\begin{mdef} \label{def:krestricted} 
%A valuation in the class of {\bf $k$-restricted  complements} is given by a hypergraph with the set of items $[m]$ as the vertex set,  where the hyperedges are restricted to be sets of size at most $k$ (including size 1). Each hyperedge 
%$e$ has a weight $w_e$. Each hyperedge corresponds to a set of perfect complements, and the bidder derives a value equal to the weight of the hyperedge only if she obtains all the items in the hyperedge. 
%The valuation  function is hence defined as  
%\[ v(S) = \sum_{e: e \subseteq S } w_e. \] 
%\end{mdef} }

To prove our results we will use the following notion of a smooth
mechanism and its corresponding implications on the price of anarchy.
\begin{mdef}[\cite{Syrgkanis2013}]\label{def:smooth-mechanism} A mechanism is $(\lambda,\mu)$-smooth
  for a class of valuations $\V=\times_i \V_i$ if for any valuation
  profile $v\in \V$, there exists a mapping $b_i': B_i\rightarrow B_i$
  such that for all $b\in \times_i B_i$:
\begin{equation}
  \sum_i u_i\left(b_i'(b_i),b_{-i}; v_i\right)\geq \lambda SW(\opt(\vp)) - \mu \sum_i P_i(b)
\end{equation}
\end{mdef}
\begin{thm}[\cite{Syrgkanis2013}] If a mechanism is
  $(\lambda,\mu)$-smooth then the price of anarchy of mixed Bayes-Nash
  equilibria of the incomplete information setting and of correlated
  equilibria in the complete information setting is at most
  $\frac{\max\{1,\mu\}}{\lambda}$
\end{thm}

\section{Price of Anarchy Upper Bounds}\label{sec:smoothness}
We will show that draft auctions are smooth mechanisms according to
Definition \ref{def:smooth-mechanism} and therefore they achieve good
social welfare at every correlated equilibrium of the complete
information setting and every mixed Bayes-Nash equilibrium of the
incomplete information setting.

For expository purposes, we begin by analyzing the case of unit-demand
bidders. In this setting, each player is allocated only one item in
the optimal allocation. To prove the smoothness property, we need to
show that from any current bid profile, every player has a deviating
strategy that depends only on his valuation and what he was doing
previously, such that either she gets utility that is a constant
fraction of his value in the optimal allocation, or his item in the
optimal allocation is currently sold at a high price.

One of the technical difficulty is that, unlike sequential item auctions,
a player is not aware, without information about other bidders'
strategies, at which step his optimal item is going to be allocated,
since this is endogenously chosen by one of his opponents. Thus,
deviations of the form: ``behave exactly as previously until the
optimal item arrives and then deviate to acquire it", are not feasible
in the case of draft auctions.\footnote{Even in the complete
  information setting, the time at which an item sells is defined by
  the strategies of other players: using this information to construct
  a deviation would not fit into the smoothness framework. In the case
  of mixed strategies, or incomplete information, the time an item
  sells is a random variable, so such a strategy is not even
  well-defined.}

Instead, our deviations for the unit-demand case have a player always
attempt to get his optimal item, while it is still available, without
changing the observed history when she loses. We show a deviation of
the following form does just that: \emph{At each time step, as long as
  your optimal item is still available, bid the maximum of your
  equilibrium bid and half your value for your optimal item. If you
  ever win, buy your optimal item.}

%\jmdelete{ The other main technical complication arises due to multi-unit
%  valuations where a player might need to acquire many items with his
%  deviation.  The moment a player's deviation is visible to other
%  players (once her winning history differs from the equilibrium
%  history), she changes the history of play and thereby faces
%  out-of-equilibrium prices, which might be arbitrarily higher. Thus
%  we need to carefully construct a deviation, such that the moment the
%  player affects the equilibrium path, she purchases all the items she
%  needs for her deviation at once. We elaborate on this difficulty
%  after completing the proof for unit-demand valuations.}
%

\begin{thm}\label{thm:unit-smooth}The draft auction for unit-demand bidders is a
  $(\frac{1}{2},2)$-smooth mechanism.
\end{thm}

%\nddelete{Before we prove this theorem, we prove a lemma.}
\hide{
\begin{lem}\label{lem:mapping-prices}
  If, for each player $i$, and every equilibrium $b$, there exists a
  strategy $b'_i$ such that

  \[u_i(b'_i(b_i), b_{-i})\geq \lambda v_i(\opti{i}) - \alpha
  p_{\opti{i}} - \beta P_i(b)\]

where $\opti{i}$ is the set of goods $i$ is allocated at $OPT$, then the
mechanism is $(\lambda, \alpha + \beta)$-smooth.
\end{lem}

\begin{proof}
  First, we note that $\opti{i}\cap \opti{i'} = \emptyset$, since no two
  players are allocated any one good at $OPT(v)$. Then, summing up the
  inequality for the deviation strategies over all players $i$, we
  have

\begin{align*}
\sum_i u_i(b'_i(b_i), b_{-i})&\geq \sum_i (\lambda v_i(\opti{i}) - \alpha p_{\opti{i}} - \beta P_i(b))\\
&\geq \lambda SW(\opt(v)) - \alpha \sum_i p_{\opti{i}} - \beta \sum_i P_i(b)\\
&\geq \lambda SW(\opt(v)) - (\alpha + \beta) \sum_i p_{\opti{i}}
\end{align*}

where the last inequality follows from the observation that the
$\opti{i}$'s are disjoint, thus implying a one-to-one mapping from
players to prices at equilibrium.

\end{proof}
Now, we prove Theorem~\ref{thm:unit-smooth}.}
\begin{proof}
  Consider a unit-demand valuation profile $v$ (i.e. $v_i(S)=\max_{j\in S}v_{ij}$) and let $j_i^*$ be the item assigned to player $i$ in
the optimal matching for valuation profile $v$. We will show that there exists a deviation mapping $b_i': B_i\rightarrow B_i$ for
each player $i$, such that for any bid profile $b$:
\begin{equation}\label{eq:udsmoothness}
  u_i(b_i'(b_i),b_{-i})\geq \frac{1}{2}v_{ij_i^*}-p_{j_i^*}(b)-P_i(b). 
\end{equation}
Consider the following $b_i'$: in every auction $t$, the player bids the
maximum of her previous bid $b_{it}$ (conditional on the history) and
$\frac{v_{ij_i^*}}{2}$, until $j_i^*$ gets
sold. If she ever wins some auction, she picks $j_i^*$.  Suppose that
$j_i^*$ was sold at some auction $t$ under strategy profile $b$. We
consider the following two cases separately, which are exhaustive
since $i$ drops out after round $t$ at most.

\begin{description} 
\item[Case 1: $i$ wins an auction $t'\leq t$ in $b_i'$. ] 
If $i$ wins with bid $b_{it'}$ then there must have been her payment under $b_i$ as well, 
and $P_i(b) = b_{it'}$. Otherwise it is $b_i^* = \frac{v_{ij_i^*}}{2}$. 
Therefore her utility is 
$$u_i(b_i',b_{-i})\geq v_{ij_i^*}-\max\left\{\frac{v_{ij_i^*}}{2},P_i(b)\right\}\geq  v_{ij_i^*} - \frac{ v_{ij_i^*} } {2} - P_i(b)
 \geq\frac{1}{2}v_{ij_i^*}-p_{j_i^*}(b)-P_i(b). $$
\item[Case 2: $i$ does not win any auction in $b_i'$. ] 
 In this case, it must be that $p_{j_i^*}(b)\geq \tfrac{1}{2}v_{ij_i^*}$ since otherwise 
$i$ would have won auction $t$. Her utility in this case  utility is zero. 
Therefore (\ref{eq:udsmoothness}) holds in this case as well. 
\end{description}

Thus we have shown that the deviation $b_i'$ always satisfies
(\ref{eq:udsmoothness}). The smoothness property follows by summing over all players and using the fact that $\sum_i
p_{j_i^*}(b)=\sum_{j\in [m]}p_j(b)=\sum_i P_i(b)$.
\end{proof}

This implies that the draft auction has Bayes-Nash and
correlated price of anarchy of at most $4$ (Theorem \ref{thm:unit-demand}). This bound is comparable
but not identical to our bound on the pure price of anarchy, which we
show to be upper-bounded by $2$ in the appendix. 

\subsection{Smoothness for constraint-homogeneous valuations }

As a next step to general subadditive valuations, we analyze
smoothness of the draft auction for a simple class of valuations.  We
subsequently show that this is the key element in proving our efficiency
results for all subadditive valuations. Specifically, we construct a
deviating strategy for the class of valuations, where each player $i$
is interested in a subset of the items $S\subseteq [m]$, and treats
all items in $S$ homogeneously, i.e. their value is a linear in the
number of items from the interest set. We will denote such valuations
as \emph{constraint-homogeneous valuations}.

\begin{mdef}[Constraint-Homogeneous Valuation]

  A valuation on a set of items is \emph{constraint-homogeneous} if it
  is defined via an interest set $S$ and a per-unit value $\hat{v}$
  such that:
\begin{equation}
\forall T\subseteq [m]: v(T)=\hat{v}\cdot |T\cap S|
\end{equation}
\end{mdef}

Unlike the unit-demand case, each player might be allocated several
units in the optimal allocation. As before, a good deviating strategy
should achieve a constant fraction of a player's valuations for her
optimal number of units, or show that the price being paid for those
units at equilibrium is high enough.  Constructing such a deviating
strategy is inherently more difficult than in the unit demand
case. The main new technical difficulty here is to construct
deviations which buy multiple units, while paying only equilibrium
prices. Once a deviation has affected the winning history, the prices
in the remaining off-equilibrium subgame are difficult to reason
about.  Thus, a player should always be trying to acquire her
optimal number of units at a good price, whilst at the same time not
changing the observed history of play.

The first idea is that the ``right price" that a player should bid to
acquire her units is half of the per-unit value, and then try to
acquire the ``right number'' of items, which is at least half the
number of units in her optimal allocation. However, consider a round
where her equilibrium bid is higher than the ``right price''. If the
bidder shades her bids down to the right price, then she may not win
that round, which changes the history for all the other players and
sets the game down an off-equilibrium path.  
%It is very hard to argue about what happens in an off-equilibrium path; the prices she now
%faces could very well be much higher than at equilibrium.  
In order to avoid this, the deviation bids the maximum of the
original bid and the right price. If the original bid is higher, she
follows the original strategy and picks the same set of items
\footnote{If the deviation were for the bidder to buy all the right
  number of units when she won because of her equilibrium bid, she
  might pay too much for them.}. If the right price is higher, she
then buys sufficient number of items to win the ``right number'' of
units, and drops out of subsequent rounds. 

The main technical meat of the paper which uses the construction of
such a deviation and forms the basis of almost all the smoothness
results in the paper is captured in the following lemma.

\begin{lem}[Core Deviation Lemma]\label{lem:multi-unit}

  Suppose that a player $i$ has a constraint-homogeneous valuation
  with interest set $S$ and per-unit value $\hat{v}$. Then in a draft
  auction there exists a deviation mapping $b_i':B_i\rightarrow B_i$
  such that, for any strategy profile $b$:
$$u_i(b_i'(b_i),b_{-i};v_i)\geq   \frac{1}{2} \frac{\hat{v}\cdot |S|}{2}-\sum_{j\in S}p_j-P_i(b).$$

\end{lem}

\newcommand{\halfs}{s^*}

The lemma is proved using the following deviation which we call the
Core Deviation.  We refer to the items in $S$ as {\em units}, and to
items not in $S$ as {\em items}.  We denote by $k_{it}$ (resp. $k_{i,<
  t}$) the number of units that player $i$ obtains in (resp. before)
auction $t$ under the original strategy $b_i$. 
We use the  shorthand notation  $\halfs:= \ceil{\frac{|S|}{2}}$

\begin{mdef}[Core Deviation]  \label{def:coredev} 
The core deviation $b'_i$ for player $i$ with a constraint-homogeneous valuation
  with interest set $S$ and per-unit value $\hat{v}$ is defined as follows.

  Let $b_i^*=\frac{\hat{v}}{2}$. In every auction $t$, she submits
  $b_{it}'=\max\left\{b_i^*, b_{it}\right\}$. If she wins with bid
  $b_i^*$, she buys $\halfs-k_{i,< t}$ units of $S$ and drops out.  If
  she wins with a bid of $b_{it}$, she buys what she did under $b_i$:
  $k_{it}$ units together with any other items she was buying under
  strategy profile $b_i$ at auction $t$. She continues to bid
  $b_{it}'$ until she acquires $\halfs$ units or the number of
  units remaining are not sufficient for her to complete
  $\halfs$ units.
\end{mdef}

The crucial observation is this: as long as the player hasn't already
acquired $\halfs$ units, she has not affected the game path created by
strategy $b_i$ in any way. From the perspective of the other bidders,
she behaved exactly as under $b_i$, by winning at her price under
$b_i$ and getting the items she would have got under $b_i$. If she
ever wins at a higher price, she acquires all the units needed to
reach $\halfs$ units in that auction and then drops out. Thus the
prices that she faces in all the auctions prior to having won $\halfs$
units are the same as the prices under strategy $b_i$.

The Core Deviation Lemma follows immediately from Lemmas \ref{lem:cd1} 
and \ref{lem:cd2}. 

\begin{lem} \label{lem:cd1} 
If player $i$ wins at least $\halfs$ units of $S$ under the Core Deviation $b_i'$ then 
$$u_i(b_i'(b_i),b_{-i}; v_i)\geq  \frac{1}{2}\halfs\hat{v}-P_i(b).$$ 
\end{lem} 
\begin{proof} 
  If player $i$ wins at least $\halfs$ units of $S$ under $b_i'$ then
  the valuation for the items she wins is at least $\halfs
  \hat{v}$.  For the auctions in which she wins with a bid of $b_{it}$
  she pays a total amount of at most $P_i(b)$ and for the (at most
  one) auction she wins with a bid of $b_i^*$ she pays at most
  $\halfs b_i^* $. So her total payment is at most
  $\halfs b_i^* + P_i(b)=\halfs
  \frac{\hat{v}}{2}+P_i(b)$.
\end{proof} 

\begin{lem} \label{lem:cd2} 

  If player $i$ wins less than $\halfs$ units of $S$ under the Core
  Deviation $b_i'$ then
\[u_i(b_i'(b_i),b_{-i};v_i)\geq   \frac{1}{2}\halfs\hat{v}-\sum_{j\in S}p_j-P_i(b).\]

\end{lem}
\begin{proof} 
  Consider the the auction under the original strategy profile $b$.  Let (by
  an abuse of notation) $p_1\leq p_2\leq \ldots \leq p_{|S|}$ be the
  prices at which the items in $S$ are sold under $b$. This is not
  necessarily the order in which they are sold.  We show in Lemma \ref{lem:cd3} that, when
  bidder $i$ wins less than $\halfs$ units under $b_i'$, it must be
  that $p_{\halfs} \geq \frac{\hat{v}}{2}$. Using this we
  obtain that
\begin{equation}\label{eq:cd2} \sum_{j\in S}p_j\geq 
  \sum_{l=\halfs}^{|S|}p_l\geq \left(|S|-\halfs+1\right)p_{\halfs}\geq
 \halfs p_{\halfs} \geq \frac{\hat{v}}{2}\halfs,
\end{equation}
where we also used the simple observation that $\halfs\leq \frac{|S|+1}{2}$.

The total payment of player $i$ under $b_i'$ in this case where she
wins less than $|S|/2$ units of $S$ is at most $P_i(b)$, therefore her
utility is (trivially) at least $-P_i(b)$. The lemma now follows from
adding the inequalities $u_i(b_i'(b_i),b_{-i};v_i)\geq-P_i(b)$ and
$0\geq \frac{\hat{v}}{2}\halfs - \sum_{j\in S}p_j$ (which holds by inequality (\ref{eq:cd2})).
\end{proof} 

\begin{lem} \label{lem:cd3} If player $i$ wins less than $\halfs$ units
  of $S$ under the Core Deviation $b_i'$ then the $\halfs$-th lowest price of the
  units in $S$ under $b$, is at least $\hat{v}/2$.
\end{lem} 
\begin{proof} 
  First, observe that if player $i$ was obtaining at least $\halfs$
  units under $b$ then she is definitely winning $\halfs$ units under
  $b_i'$, since she is always bidding at least as high.  So, we can assume that under $b$ player $i$ wins fewer than
  $\halfs$ units.

  Recall that $p_1\leq p_2\leq \ldots \leq p_{|S|}$ are the prices at
  which the units in $S$ are sold under $b$. Let $P_t$ be the price
  of auction $t$ (under $b$).  Let $t^*$ be the first auction that
  was won at price $P_{t^*}\leq p_{\halfs}$ under $b$ but not by
  bidder $i$.  We know that such an auction must exist; under
  $b$ there are $\halfs$ units of $S$ that are sold at a price at
  most $p_{\halfs}$, and since player $i$ wins less than $\halfs$ of
  them, some of them are not won by player $i$.

  We now argue that player $i$ is still bidding in auction $t^*$ under
  $b'_i$.  First of all, she has not won $\halfs$ units prior to $t^*$.
  The other condition needed for her to be active is that there are at
  least $\halfs-k_{i,<t^*}$ units available for sale in that
  auction. This follows from the fact that for any auction $t<t^*$ for
  which $P_t\leq p_{\halfs}$, we know that player $i$ was winning under
  $b_i$. Thus every unit that was sold prior to $t^*$ at a price of
  less than or equal to $p_{\halfs}$ was sold to player $i$. There are
  $\halfs$ units sold at a price $\leq p_{\halfs}$ and the number of
  such units sold prior to $t^*$ is at most the number of total units
  won by bidder $i$ prior to $t^*$. Thus the number of available units
  available at $t^*$ is at least: $\halfs-k_{i,<t^*}$.

  Finally, we argue that $P_{t^*}\geq b_i^*$. Suppose for the sake of
  contradiction that $P_{t^*}<b_i^*$. Then player $i$ wins auction
  $t^*$. Since she was not winning $t^*$ under $b_i$, it must be that
  she is winning $t^*$ with a bid of $b_i^*$. Thus in that auction she
  will buy every unit needed to reach $\halfs$ units. By the analysis
  in the previous paragraph, we know that there are still enough units
  available for sale to reach $\halfs$. Thus in this case she will win
  $\halfs$ items, a contradiction with the main assumption of the Lemma.  Therefore, $ b_i^*\leq P_{t^*}$ and
  by definition, $P_{t^*}\leq p_{\halfs} $ and $b_i^* = \frac{\hat{v}}{2}$.
\end{proof} 

An easy corollary of the above core deviation lemma is that when all
players have constraint-homogeneous valuations, the draft auction is a
$\left(\frac{1}{4},2\right)$-smooth mechanism, and thus has a price of
anarchy of at most 8 for these valuations.
\begin{cor}\label{cor:smooth-constraint}
  The draft auction is a $(\frac{1}{4},2)$-smooth mechanism when
  bidders have constraint-homogeneous valuations. (proof in Appendix \ref{sec:omitted-smoothness})
\end{cor}

\subsection{Extension to more general valuations} 
\label{sec:extension}
We will next show that smoothness for constraint-homogeneous
valuations implies smoothness for a much larger class of
valuations. We achieve this based on the following re-interpretation
of the results in \cite{Syrgkanis2013}\footnote{\cite{HartlineLectures} gives a special case
of this re-interpretation for the mechanism defined by simultaneous single-item auctions, showing how smoothness for additive valuations implies smoothness for unit-demand (and XOS) valuations}.

\begin{mdef}[Pointwise Valuation Approximation]\label{def:approximation} A
  valuation class $V$ is pointwise $\beta$-approximated by a valuation
  class $V'$, if for any valuation profile $v\in V$, and for any set
  $S\subseteq [m]$, there exists a valuation profile $v'\in V'$ such
  that: $\beta v'(S)\geq v(S)$ and for all $T\subseteq [m]$: $v(T)\geq
  v'(T)$.
\end{mdef}

Note that, importantly, the valuation $v'$ can depend on $S$. $\beta
v'$ only needs to upper bound $v$ at $S$, while $v'$ needs to lower
bound $v$ everywhere else. This is much weaker than the related notion
of approximation by a function class, where for every $v$ we ask for a
single $v'$ such that $v$ is sandwiched between $\beta v'$ and $v'$
everywhere.

\begin{lem}[Extension Lemma] If a mechanism for a combinatorial
  auction setting is $(\lambda,\mu)$-smooth for the class of
  valuations $V'$ and $V$ is pointwise $\beta$-approximated by $V'$, then
  it is $\left(\frac{\lambda}{\beta},\mu\right)$-smooth for the class
  $V$.
\end{lem}

\begin{proof}
  Consider a valuation profile $v$ where each valuation comes from
  valuation class $V$ . For each player $i$ let $S_i^*$ be her optimal
  allocation under $v$ and let $v^*$ be the valuation profile such
  that $v_i^*\in V'$ is the valuation that $\beta$-dominates $v_i$ for
  set $S_i^*$: i.e. $\beta\cdot v_i^*(S_i^*)\geq v_i(S_i)$ and for all
  $T\subseteq [m]$: $v_i(T)\geq v_i^*(T)$. By the first property we
  get that $\beta\cdot SW(\opt(v^*))\geq SW(\opt(v))$. By the second property
  we get that for all bid profiles $b$: $u_i(b; v_i)\geq u_i(b;
  v_i^*)$. Let $b_i':B_i\rightarrow B_i$ be the deviation mapping that
  is designated by the smoothness property of the mechanism under
  $v^*$. Then for any bid profile $b$:
\begin{align*}
  \sum_i u_i(b_i'(b_i),b_{-i};v_i)\geq~& \sum_i u_i(b_i'(b_i),b_{-i}; v_i^*)\geq \lambda SW(\opt(v^*))-\mu \sum_i P_i(b)\\
\geq~& \frac{\lambda}{\beta} SW(\opt(v))-\mu \sum_i P_i(b)
\end{align*}
which implies the mechanism is smooth for the valuation class $V$.
\end{proof}

\paragraph{Identical Items and Concave Symmetric Valuations.}

We first consider the case where all items are identical and players
have a valuation that is a concave function of the number of items
acquired, i.e., $v_i(S)=f_i(|S|)$ for some non-decreasing concave
function $f_i: \mathbb{N}\rightarrow \mathbb{R}^+$.  We call these as
{\em concave symmetric valuations}.  We show that all such valuations
can be pointwise $1$-approximated by constraint-homogeneous valuations.
As a corollary we get that the price of anarchy of draft auctions for
this case is at most $8$ (Theorem \ref{thm:symmetric}).

\begin{thm}\label{thm:multi-unit} 
  The class of concave symmetric valuations is pointwise 1-approximated
  by constraint-homogeneous valuations (proof in Appendix \ref{sec:omitted-smoothness}).
\end{thm}

%Thus we conclude that for the case of identical items, the price of anarchy of draft auctions with concave symmetric valuations
%is at most $8$.

\paragraph{Heterogeneous items.}
We next turn to simplest class of valuations over heterogeneous items,
additive valuations. We show that the {Core Deviating Lemma} implies a
$O(\log(m))$ price of anarchy. The key technical step is showing that
any additive valuations can be pointwise approximated within a
logarithmic factor by a {constraint-homogeneous valuation}, via a
standard bucketing argument.

\hide{\nddelete{Specifically, we show that for any additive valuation
$v(\cdot)$ and for any set $S\subseteq [m]$, there exists a
constraint-homogeneous valuation $\hat{v}$ such that $v(S)\leq
O(\log(m))\cdot \hat{v}(S)$ and for any $T\subseteq [m]$:
$\hat{v}(T)\leq v(T)$.
The idea is that we can group items in $S$, such that the valuation of
a player for items within each group are within a factor of $2$ of
each other, grouping together all items with value less than a $1/m$
fraction of the highest valued item. Next we observe that there are
$O(\log(m))$ such groups and thereby one of these groups must yield a
total value to the bidder of at least a logarithmic fraction of her
initial valuation for the whole set. Moreover, this group cannot be
the last group of small items, since in total they are less valuable
than the highest valued item. Thus we can consider the constraint
homogeneous valuation whose target set is the most valuable among the
above groups. Treating the items within each group as identical and
with value equal to the minimum valued item in the group, can only
result in a constant factor difference from the initial valuation.}
}

\begin{lem} \label{lem:additive} 
Additive valuations can be pointwise $2(\log(m-1)+1)$-approximated
  by constraint-homogeneous valuations. (proof in Appendix \ref{sec:omitted-smoothness})
\end{lem}

Combining the latter lemma with the smoothness of draft auctions for
constraint-homogeneous valuations we get the efficiency guarantee for
additive valuations.

\begin{cor}\label{thm:additive}
  The draft auction is a $(\frac{1}{8(\log(m-1)+1)},2)$-smooth mechanism for
  additive bidders, implying a price of anarchy of at most $16(\log(m-1)+1)$.
\end{cor}

Additionally, by the definition of XOS valuations, it is easy to see
that they are pointwise $1$-approximated by additive valuations, in the
sense of Definition \ref{def:approximation}. Moreover, it is known
(see e.g. \cite{Bhawalkar2011}) that subadditive valuations can be pointwise
$H_m$-approximated by additive valuations, in the sense of Definition
\ref{def:approximation}. This leads to the following two corollaries.

\begin{cor}\label{cor:xos} The draft auction is a
  $(\frac{1}{8(\log(m-1)+1)},2)$-smooth mechanism for XOS valuations
  and it is a $\left(\frac{1}{8 H_m(\log(m-1)+1)},2\right)$-smooth mechanism for
  subadditive valuations.
\end{cor}

This in turn implies that the price of anarchy of draft auctions is
$O(\log(m))$ for XOS valuations (Theorem \ref{thm:xos}) and $O(\log^2(m))$ for subadditive valuations (Theorem \ref{thm:subadditive}).

\section{Related Issues and Conclusion}

\paragraph{Instances of sequential item auctions.} Here we identify
natural candidates for implementing a draft auction in place of a
sequential item auction.  Sequential item auctions are used by auction
houses such as Sotheby's and Christie's, which auction off art,
jewelry, wine, etc.  The United States government auctions off a whole
bunch of seized and surplus merchandise ranging from electronics and
automobiles to industrial equipment and real estate \citep{usa}.
Another notable instance of a sequential item auction is the
auctioning of players to teams in a professional league such as the
{\em Indian Premiere League} \citep{hindu, youtube, cricinfo}.  We believe that
in many of these cases switching to a draft auction would be easy and
beneficial.  RTC auctions have already been used in some instances and
have been found to give a higher revenue, for example condo sales in
Miami \citep{ashenfelter} and selling water rights in Chile
\citep{alevy}.

\paragraph{Why first price?} 
The auction in each round of the draft auction is a sealed bid first
price auction.  Our results continue to hold for second price auctions
under an extra assumption of ``no overbidding".  Although overbidding
(i.e., bidding above one's valuation) seems unnatural and unhelpful,
one cannot easily rule out such strategies in a second price
auction. This makes analyzing second price auctions much harder, and
the no overbidding assumption has become a common way around this
difficulty \cite{}.  One exception is the unit-demand case where we
can show that overbidding is a dominated strategy.

Ascending price auctions hold additional difficulties since in this
case, each round is itself a sequential game. 
%and bidders can base their strategies also on the sequence of increasing bids and the identities of the bidders. \ndcomment{I am not completely happy with   this statement.} 
Due to this, we cannot hide our deviations until the very end and win a whole bunch of items once our 
deviation is apparent to others, like we do now. Even within a round, as soon as it becomes clear that 
we are deviating from the equilibrium, other players may change their behaviour before we can win the round. 
% the deviations we construct, which 
Nonetheless, we believe that our bounds should hold ``in
principle'' for these auctions as well, and resolving whether they do
is an interesting open question.

%\ndcomment{Have I said something incorrect here? } 

\paragraph{Why social welfare? }
We picked social welfare as the objective in this paper.  Social
welfare is probably the most common objective in the study of
combinatorial auctions, and is well motivated when the
auctioning authority is something like the government.  Folklore has
it that social welfare is also the ``right'' objective in the absence
of a monopoly, that is if similar items can be obtained by other
sellers as well.

Another natural objective is the revenue from the auction.
As mentioned in the section on related work, 
experimental results indicate that the revenue from draft auctions is
higher than other formats such as sequential item auctions on real world instances.  
Theoretical analysis of revenue seems more difficult as is evidenced by the dearth of such results. 
One difficulty is, unlike social welfare
which only depends on the allocation, the revenue depends on the
payments as well and therefore there is no clear benchmark for revenue
as an objective. 
% In previous work revenue comparisons have been done
%for several specific distributions with two superadditive bidders and
%two items, and show that the expected revenue is the same\citep{harstad}.  
We can answer simple questions about revenue, such
as ``Is the revenue from one auction instance-wise better than the
other?"  The answer is, no, for the complete information
case. Resolving this question for the Bayesian case for reasonable
distributions such as regular or monotone hazard rate distributions
and analysing the revenue of these auctions in general is also an
important direction for future research.

\paragraph{Sequential  vs. simultaneous auctions.}
Another simple auction is a simultaneous item auction, where bidders
submit sealed bids for all the items simultaneously and each item is
sold to the highest bidder. \citet{Feldman2013} have shown a constant price of
anarchy for simultaneous auctions for subadditive valuations, which
indicates that simultaneous auctions are better than sequential
auctions.  Sequential item auctions still seem to be quite commonly
preferred over simultaneous item auctions in practice, and we feel
draft auctions may be better suited than simultaneous item auctions
for many of these scenarios.  A better theoretical understanding of
the advantages of each and direct comparisons between the two would be
very valuable.

\paragraph{Other open problems.} Our work raises several open questions, the most intriguing one being
whether the price of anarchy for subadditive valuations is at most a constant. 
A constant upper bound on subclasses such as gross substitutes or submodular valuations would also be very interesting. 
Can we show any upper bounds for classes beyond subadditive valuations? 
A natural candidate is the class of valuations with restricted complements, introduced by \citet{Abraham2012}.

%However, we note that when the number of items is large, conducting a simultaneous auction could be very challenging 
%and would need extra infrastructure, such as a suitable software support. It could also be challenging for the bidders to keep 
%track of all the auctions at the same time. The simplicity of sequential auctions and the fact that they  need almost no 
%infrastructure is the reason they are used so widely. 

%\ndcomment{Need to mention existence results somewhere} 

%\section{Bidders with Complementarities}
%\input{complements}

\bibliographystyle{plainnat}
\bibliography{simple-auctions}

\begin{appendix}

% !TEX root=draft.tex
\section{Valuation Classes}\label{app:valuations}
\begin{mdef} \label{def:submodular} 
A monotone valuation function is {\bf submodular} if it exhibits the diminishing marginal value property, which to be precise is that 
\[ \forall~ S\subseteq T, \forall~ i \notin T, v(S \cup \set{i} ) - v(S) \geq 
v( T \cup \set{i} ) - v(T ) . \]
\end{mdef} 

\begin{mdef} \label{def:gs} 
Valuations with the {\bf gross substitutes} property are defined in terms of the corresponding {\em demand} function.  Given prices $p_j$ for all $j\in [m]$, the demand correspondence is 
\[ x(p_j)_{j\in [m]} := \arg \max_{S \subseteq [m]} \set{ v(S) - \sum_{j\in S} p_j }. \] 
A demand function satisfies gross-substitutes if increasing the price of one item does not decrease the 
demand for any other item.  If the demand function is a correspondence, then it satisfies the gross-substitute condition when the following holds: if an item $j$ is in some demand set under price $p=(p_1,\ldots, p_m)$, then after increasing the price of item $j$ and keeping the rest of the prices the same,  there exists a demand set under the new prices that contains $j$.
\end{mdef}

\section{Definition of Correlated Equilibrium} 

\begin{mdef} {\bf Correlated equilibrium} \label{def:correlated} A
  correlated equilibrium is a distribution $X$ over joint strategy
  profiles such that, for each player $i$, following the suggestion
  $s_i$ drawn from the distribution $X$ is a best-response, in
  expectation over the suggestions $s_{-i}$, not known to $i$ and
  assuming everyone else plays according to their suggestion:

\[\mathbb{E}_{\s_{-i},\v}[u_i(\s(\v)) \mid s_i] \geq \mathbb{E}_{s_{-i},\v}[u_i\left(s'_i(v_i), \s_{-i}(\v_{-i})\right) \mid s_i]\]

Note that the deviation is allowed to depend on the suggestion (in the
event that $s'_i$ is required to be independent of $s_i$ for all $i$,
we call $s$ a coarse correlated equilibrium).
\end{mdef}

\section{Omitted Proofs}

\subsection{Proofs from Section \ref{sec:smoothness}}\label{sec:omitted-smoothness}

\begin{proofof}{Theorem \ref{thm:multi-unit}}
%\begin{proof}

  Consider a valuation profile $v$ as described in the theorem
  (i.e. $v_i(S)=f_i(|S|)$).  Consider a set $S\subseteq [m]$ and let
  $v_i'$ be the constraint-homogeneous valuation with interest set $S$
  and per-unit valuation $\hat{v}_i' =
  \frac{f_i(|S_i^*|)}{|S_i^*|}$. By concavity of the valuation $v_i$
  we have that for any $T\subseteq [m]$:
 \begin{equation}
v_i'(T)= \hat{v}_i'\cdot |T\cap S_i^*|= \frac{f_i\left(|S_i^*|\right)}{|S_i^*|}\cdot |T\cap S_i^*|\leq f_i\left(|T\cap S_i^*|\right) \leq v_i(T)
\end{equation}
Additionally, $v_i'(S_i^*)=f_i(|S_i^*|)=v_i(S_i^*)$.
%\end{proof}

\end{proofof}

\begin{proofof}{Corollary \ref{cor:smooth-constraint}}
Consider a constraint-homogeneous valuation profile $v$ and a bid
  profile $b$. Let $S_i^*$ be the units allocated to player $i$ in the
  optimal allocation for profile $v$. Also let $S_i$ be the interest
  set of each player and $\hat{v}_i$ his per-unit value. Consider the
  alternative valuation profile where each player $i$ has a
  constraint-homogeneous valuation $v_i'$ with interest set
  $S_i'=S_i\cap S_i^*$ and per unit value $\hat{v}_i'=\hat{v}_i$.

Observe that for any $T\subseteq [m]$, $v_i(T)\geq v_i'(T)$ and
$v_i(S_i^*)=v_i'(S_i^*)$. Thus, for any bid profile $b$:
$u_i(b;v_i)\geq u_i(b;v_i')$ and $SW(\opt(v'))\geq SW(\opt(v))$. Invoking
Lemma \ref{lem:multi-unit} on valuations $v_i'$, we get that there
exists a deviation mapping $b_i':B_i\rightarrow B_i$ for each player
$i$ such that for any strategy profile $b$:
\begin{equation*}
  \sum_i u_i(b_i'(b_i),b_{-i};v_i)\geq  \sum_i u_i(b_i'(b_i),b_{-i};v_i') \geq \frac{1}{4}\opt(v')-2\sum_i P_i(b) \geq \frac{1}{4} \opt(v) - 2\sum_i P_i(b),
\end{equation*}
where we have once again used the fact that $\sum_i
p_{j_i^*}(b)=\sum_{j\in [m]}p_j(b)=\sum_i P_i(b)$.
%  Consider a constraint-homogeneous valuation profile $v$ and a bid
%  profile $b$. Let $S_i^*$ be the units allocated to player $i$ in the
%  optimal allocation for profile $v$. Also let $S_i$ be the interest
%  set of each player and $\hat{v}_i$ his per-unit value. Consider the
%  alternative valuation profile where each player $i$ has a
%  constraint-homogeneous valuation $v_i'$ with interest set
%  $S_i'=S_i\cap S_i^*$ and per-unit value $\hat{v}_i'=\hat{v}_i$.
%
%Observe that for any $T\subseteq [m]$, $v_i(T)\geq v_i'(T)$ and
%$v_i(S_i^*)=v_i'(S_i^*)$. Thus, for any bid profile $b$:
%$u_i(b;v_i)\geq u_i(b;v_i')$ and $SW(\opt(v'))\geq SW(\opt(v))$. Invoking
%Lemma \ref{lem:multi-unit} on valuations $v_i'$, we get that there
%exists a deviation mapping $b_i':B_i\rightarrow B_i$ for each player
%$i$ such that for any strategy profile $b$:
%\begin{equation*}
% u_i(b_i'(b_i),b_{-i};v_i)\geq   u_i(b_i'(b_i),b_{-i};v_i') \geq \frac{1}{4}v'(\optpi{i})-\sum_i P_i(b) - \sum_i p_i(\optpi{i}) \geq \frac{1}{4}v(\opti{i})-\sum_i P_i(b) - \sum_i p_i(\opti{i}) 
%\end{equation*}
%
%where $\optpi{i}$ denotes the set of items $i$ wins at $OPT$ according
%to $v'$. Applying Lemma~\ref{lem:mapping-prices}, the lemma follows.
%\ndcomment{Replace orginal proof.}
\end{proofof}

\begin{proofof}{Lemma \ref{lem:additive}}
  Consider an additive valuation $v$, i.e $v(T)=\sum_{j\in T}v_{j}$. Let $S$ be a set of $k$ items and sort the items in $S$ in 
decreasing order of value: $v_{1} \geq v_{2} \geq \ldots \geq
  v_{k}$. Consider the partition of items $\mathcal{P}$ where $I_1
  = \left\{j : v_{j} \geq \frac{v_{1}}{2}\right\}$, and more
  generally, for any $t\in \left[2, \log(k-1)\right]$

 \[I_t = \left\{j \mid \frac{v_{1}}{2^{t-1}} > v_{j} \geq
     \frac{v_{1}}{2^t}\right\}.\]

   Let the final set $I_f$ contain all the smallest items, $I_f =
   \left\{j : v_{j} < \frac{v_{1}}{k-1}\right\}$. Notice that the
   largest-valued item in $I_f$ has value at most $\frac{v_{1}}{k-1}$
   and there are at most $k-1$ items in $I_f$, thus, $v_i(I_f) <
   v_{1}$, and so $v_i(I_1) > v_i(I_f)$.  There are $\log(k-1)+1$
   sets in $\mathcal{P}$, so the largest valued one has value at least
   $\frac{v(S)}{\log(k-1)+1}$.  It cannot be  $I_f$ so it is one of the first $\log(k-1)$ sets.
Thus if we denote with $\tau=\arg\max_{t\in [1,\ldots,f-1]}v_i(I_t)$ we get that:
\begin{equation}
v(I_{\tau}) \geq \frac{v(S)}{\log(k-1)+1}
\end{equation}
Now consider the constraint-homogeneous valuation $v_i'$ with interest set $I_{\tau}$ and
$\hat{v}'=\min_{j\in I_\tau}v_{j}$. It is obvious that for any set $T\subseteq [m]$: $v(T)\geq v'(T)$,
since an element's valuation was either set to zero or decreased under $v'$. Additionally, since items in $I_\tau$ only differ
by a factor of $2$, we also get that $v'(S)=\hat{v}\cdot|I_\tau|\geq
\frac{v(I_\tau)}{2} \geq \frac{v(S)}{2(\log(k-1)+1)}$.
\end{proofof}

\subsection{Existence of Pure SPE for Single-Item Draft Auctions}
\begin{obs}\label{obs:existence}
  Single draft auctions always have pure subgame perfect equilibria,
  where bidders do not use weakly dominated strategies, for bidders
  who have arbitrary valuations.
\end{obs}
\begin{proof}
  By Theorem 2.1 of previous work~\cite{Syrgkanis2012}, every first-price
  single-item auction with externalities has a pure Nash equilibrium
  which doesn't use dominated strategies. Single draft auctions can be
  thought of as single auctions with externalities: bidder $i$ has an
  associated value in the remaining game as a function of the player
  $j$ who wins the current auction (since $j$ has a well-defined best
  item to take once he wins). Thus, one can construct a subgame
  perfect equilibrium by backwards induction. The final auction has no
  externalities. As a function of who wins the $k$th auction and what
  they take, there is a well-defined value each bidder has for the
  remaining auctions. Those define the externalities for the $k$th
  auction.
\end{proof}

\begin{obs}
  The above proof relies on the single selection process: if a player
  can select more than one item, the set she chooses (and thus her
  externalities for winning, and other's externalities for her
  winning) change as a function of the price at which she wins.
\end{obs}

\subsection{The Pure PoA for unit-demand bidders is at most $2$}
We give a tighter upper bound on the price of anarchy for pure Nash
equilibria, which was the second part of Theorem
\ref{thm:unit-demand}.
\begin{obs}
  The pure-Nash PoA is 2 for unit-demand bidders.
\end{obs}

\begin{proof}
  Consider an agent $i$ who gets the item $\jopt{i}$ in $OPT$. There are
  2 cases: that $i$ wins in a round where $\jopt{i}$ has yet to be sold,
  and where $i$ wins a round after $\jopt{i}$ has been sold. In case $1$,
  $v_{i,j(i)} \geq \valopt{i}$ (since $i$ has the choice to take
  $\jopt{i}$). In case 2, $v_{i,j(i)} \geq \valopt{i} -
  p(\jopt{i})$. Summing up over all players, we have $SW_{EQ} \geq
  SW_{OPT} - REV_{EQ}$. Thus, since $REV_{EQ} \geq SW_{EQ}$, by
  individual rationality, $2SW_{EQ} \geq SW_{OPT}$.
\end{proof}

\subsection{Inefficiency:  Proof of Theorem ~\ref{thm:lowerbounds}}

In this section, we show that draft auctions for unit demand bidders
may be inefficient, even when all players agree on the relative
ordering of the items by their value. We also show that the pure price
of anarchy is between $2$ and $209/177 > 1.22$ for unrestricted
unit-demand bidders. Here, we present an example of the inefficiency
which arises from the competition between agents.

The inefficiency of arbitrary subgame perfect equilibria suggests a
new conjecture: is the price of stability $1$? Our example also shows
the price of stability is strictly larger than $1$, even when all
bidders share the same ranking of the items by value. Indeed, the
example we give below also shows a setting \emph{in which no
  equilibrium is optimal}.

\begin{lem}\label{lem:stability}
  The price of stability for unit-demand bidders in draft auctions is
  strictly larger than $1$, even when all bidders share the same
  ordering of items by value. This implies that the price of anarchy
  is also greater than $1$.
\end{lem}

\begin{figure}
\centering
\input{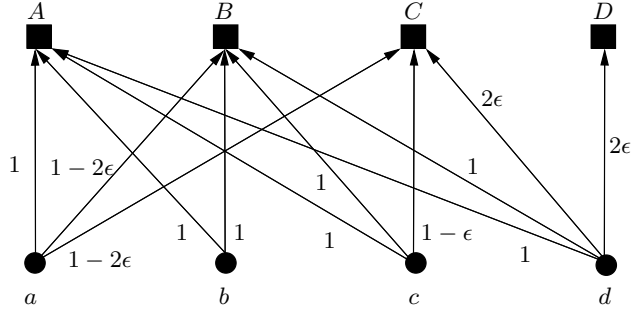}
\caption{An example of inefficiency in draft auctions for unit-demand bidders}\label{fig:inefficient}
\end{figure}

\begin{proof}
The example shown in Figure~\ref{fig:inefficient} has no equilibrium
where the optimal allocation is given. The optimal allocation is
$(a,A), (b,B), (c,C), (d,D)$, with a total weight of
$3+\epsilon$. Since the agents have the same ordering on items, item
$a$ will be selected first. So, for the optimal allocation to be an
equilibrium, $a$ must win the first round.  Then, assuming in the
subgame where $(a,A)$ is removed that the optimal allocation is
reached and solving for prices, the price vector for such an
allocation is $(1-2\epsilon, 1-2\epsilon, 0,0)$, where player $b$
price-sets in the first round and player $d$ price-sets in the second round\footnote{ We can ensure $b$ is the price-setter for player $a$ by
  slightly increasing the weight on the edges $(b,A)$. That way, $b$
  will have a slightly higher incentive to beat $a$ than $d$ has.}.
Thus, player $a$ gets utility $2\epsilon$ from winning the first
round.

Now, consider what happens if $a$ loses the first round to player $b$. If we
assume in any subgame, the optimal allocation is the one which is
made, the allocation in the subgame will be $(a,C), (c,B),
(d,D)$. Then, $c$ will pay $1-2\epsilon$, and $a$ will pay nothing,
for utility $1-2\epsilon$. So, if $a$ chooses to lose to player $b$,
her utility is $1-4\epsilon$ higher than if she won the first round.

\end{proof}

\begin{lem}\label{lem:lower}
  The pure price of anarchy of draft auctions for unit-demand bidders
  is at least 1.22.
\end{lem}

\begin{proof}

Consider the matrix

\[ \left( \begin{array}{c|ccc}
A & 32 & 31 & 83 \\
B & 9 &  84&  97\\
C & 2 & 42 & 93 \end{array} \right)\] 

and the strategy profile where $B$ wins first, supported by price
$p_1 = 51$ by agent $A$, then $C$ wins at price $p_2 = 0$, then $A$
wins at price $p_3 = 0$. The allocation will be $(1, 3, 2)$, with
social welfare $97 + 32 + 42 = 171$, while the optimal allocation
$(1,2,3)$ gives social welfare of $32 + 84 + 93 = 209$. Thus, the PoA is
at least $209/171 > 1.22$.

Showing this is a SPE is not difficult: it is necessary, however, for
$A$ to be the price supporter in round 1, rather than $C$ (both have
equal externality for $B$ winning round 1). if $C$ price-sets, $B$
would rather lose round $1$ and the outcome will be efficient.
\end{proof}

The proof of Theorem~\ref{thm:lowerbounds} follows from Lemma~\ref{lem:stability} and Lemma~\ref{lem:lower}.

\subsection{Non-uniqueness of equilibria and Non-dominance of  Revenue}

We show that equilibria are non-unique for draft auctions, even with
unit-demand bidders. Consider the following valuations, where $a,b,c$
are the bidders, and the items are $A^*, B^*$. Suppose $a$ has value
$1$ for $A^*$ and $0$ for $B^*$, and $b$ and $c$ have value $2$ for
either item. Then, there are two pure equilibria: one where $b$ wins
$A^*$ for price $1$ and $c$ wins $B^*$ for $0$, and the other where
$c$ wins $A^*$ for $1$ and $b$ wins $B^*$ for $0$.

This example also shows that item auctions have multiple equilibria
(if $A^*$ is sold first, either $b$ or $c$ winning at price $1$ and
then the other winning the second round at price $0$ is an
equilibrium). If we consider the item auction with the reverse order,
where $B^*$ is sold and then $A^*$, the price in each round will be
$1$, which shows the revenue from the best order is better than the
revenue from draft auctions.

\newpage
\end{appendix}

\end{document}